\documentclass[showkeys]{revtex4}

\newif \iffinal \finalfalse

\usepackage{algorithm}
\usepackage{algorithmic}
\usepackage[mathscr]{eucal}
\usepackage{epsfig}
\usepackage{amssymb}
\usepackage{amsfonts}
\usepackage{amsmath}
\usepackage{latexsym}
\usepackage{bm}
\usepackage{xspace}
\usepackage{boxedminipage}
\usepackage{synttree}

\newcommand{\qed}{\hfill\rule{7pt}{7pt}}
\newtheorem{theorem}{Theorem}[section]
\newtheorem{lemma}[theorem]{Lemma}
\newtheorem{observation}[theorem]{Observation}

\newenvironment{proof}{\noindent{\bf Proof:}}{\qed\newline\medskip}

\newtheorem{definition}[theorem]{Definition}
\newtheorem{fact}[theorem]{Fact}

\numberwithin{equation}{section}

\newcommand{\fousumeltbl}{\hat{f}(S)}
\newcommand{\alg}{\mathcal{A}}

\newcommand{\cc}{\mathfrak{C}}
\iffinal
\newcommand{\rnote}[1]{}
\newcommand{\anote}[1]{}
\newcommand{\remove}[1]{}
\else
   \newcommand{\rnote}[1]{\footnote{{\bf [[Rocco: {#1}\bf ]] }}}
   \newcommand{\anote}[1]{\footnote{\bf [[Alp: {#1}\bf ]]}}
   \newcommand{\remove}[1]{\par $<<${\it removed part}$>>$}
   \newcommand{\old}[1]{}
\fi

\newcommand{\dist}{\mathscr{D}}
\newcommand{\bj}{\text{FS}}
\newcommand{\ex}{\mathsf{EX}}
\newcommand{\qex}{\mathsf{QEX}}
\newcommand{\qmq}{\mathsf{QMQ}}
\newcommand{\dnf}{\mathsf{DNF}}
\newcommand{\mq}{\mathsf{MQ}}
\newcommand\conj[1]{\overline{#1}}

\newcommand{\ignore}[1]{}
\newcommand{\bits}{\{-1,1\}}

\newcommand{\bn}{\bits^n}

\newcommand{\R}{{\bf R}}
\newcommand{\Inf}{\mathrm{Inf}}

\newcommand{\E}{{\bf E}}

\renewcommand{\P}{\mathbf{Pr}}
\newcommand{\eps}{\epsilon}

\newcommand{\sgn}{\mathrm{sgn}}

\newcommand\innp[2]{\langle {#1},{#2} \rangle}
\newcommand{\fisafunc}{f \colon \bn \to \bits}

\begin{document}
\title{Quantum Algorithms for Learning and Testing Juntas}
\date{July 20, 2007}

\author{Alp At\i c{\i}}
\thanks{Work done while at the Department of
Mathematics, Columbia University, New York, NY 10027}
\email{alpatici@gmail.com}
\affiliation{
Citadel Investment Group\\
Chicago, IL 60603}

\author{Rocco~A. Servedio}
\thanks{Supported in
part by NSF award CCF-0347282, by NSF award
CCF-0523664, and by a Sloan Foundation Fellowship.}
\email{rocco@cs.columbia.edu}
\affiliation{
Department of Computer Science\\
Columbia University\\
New York, NY 10027}

\begin{abstract}
In this article we develop quantum algorithms for learning and
testing \emph{juntas}, i.e. Boolean functions which depend only on
an unknown set of $k$ out of $n$ input  variables. Our aim is to
develop efficient algorithms:
\smallskip

\textbullet \quad whose sample complexity has no dependence on $n$, the dimension
of the domain the Boolean functions are defined over;
\smallskip

\textbullet \quad with no access to any classical or quantum membership (``black-box'')
queries.  Instead, our algorithms use only classical examples
generated uniformly at random and fixed quantum superpositions of
such classical examples;
\smallskip

\textbullet \quad which require only a few quantum examples but possibly
many classical random examples (which are considered quite ``cheap''
relative to quantum examples).
\bigskip

Our quantum algorithms are based on a subroutine $\bj$ which enables
sampling according to the Fourier spectrum of $f$; the $\bj$
subroutine was used in earlier work of Bshouty and Jackson on
quantum learning. Our results are as follows:
\smallskip

        \textbullet \quad We give an algorithm for testing $k$-juntas to accuracy $\eps$
        that uses $O(k/\epsilon)$ quantum examples.  This improves
        on the number of examples used by the best known classical
        algorithm.
\smallskip

        \textbullet \quad We establish the following lower bound: any $\bj$-based $k$-junta testing algorithm requires $\Omega(\sqrt{k})$ queries.
\smallskip

        \textbullet \quad We give an algorithm for learning $k$-juntas to accuracy $\eps$ that uses
        $O(\epsilon^{-1} k\log k)$ quantum examples and $O(2^k \log(1/\eps))$ random examples.
        We show that this learning algorithms is close to optimal by giving a related lower bound.
\smallskip

\end{abstract}

\pacs{03.67.-a, 03.67.Lx}
\keywords{juntas, quantum query algorithms, quantum property testing, computational learning theory, quantum computation, lower bounds}

\maketitle

\section{Introduction}
\subsection{Motivation}

The field of \emph{computational learning theory} deals with the
abilities and limitations of algorithms that learn functions from
data.  Many models of how learning algorithms access data have been
considered in the literature. Among these, two of the most prominent
are via \emph{membership queries} and via \emph{random examples}.
Membership queries are ``black-box'' queries; in a membership query,
a learning algorithm submits an input $x$ to an oracle and receives
the value of $f(x)$. In models of learning from random examples,
each time the learning algorithm queries the oracle it receives a
labeled example $(x,f(x))$ where $x$ is independently drawn from
some fixed probability distribution over the space of all possible
examples. (We give precise definitions of these, and all the
learning models we consider, in Section~\ref{sec:jprelim}.)

In recent years a number of researchers have considered {quantum}
variants of well-studied models in computational learning theory,
see e.g. \cite{AKMPY,AS05,BSHJA,C06,HMPPR,IKRY,RSSG}. As we describe
in Section~\ref{sec:jprelim}, models of learning from quantum
membership queries and from fixed quantum superpositions of labeled
examples (we refer to these as \emph{quantum examples}) have been
considered; such oracles have been studied in the context of
\emph{quantum property testing} as well \cite{BFNR,FMSS,MN}. One
common theme in the existing literature on quantum computational
learning and testing is that these works study algorithms whose only
access to the function is via some form of quantum oracle such as
the quantum membership oracle or quantum example oracles mentioned
above.  For instance, \cite{BSHJA} modifies the classical Harmonic
Sieve algorithm of \cite{JACKSON} so that it uses only uniform
quantum examples to learn $\dnf$ formulas. \cite{BFNR} considers the
problem of quantum property testing using quantum membership queries
to give an exponential separation between classical and quantum
testers for certain concept classes. \cite{AS05} studies the
information-theoretic requirements of exact learning using quantum
membership queries and Probably Approximately Correct (PAC) learning
using quantum examples. Many other articles such as
\cite{RSSG,AKMPY,HMPPR} could further extend this list.

    As the problem of building large scale quantum computers remains a major challenge, it is natural to
    question the technical feasibility of large scale implementation of the quantum oracles considered in the literature.
    It is desirable to minimize the number of quantum (as opposed to classical) oracle queries or examples
    required by quantum algorithms.
    Thus motivated, in this paper we are interested in designing testing and learning
    algorithms with access to both quantum and classical sources of information
(with the goal of minimizing the quantum
    resources required).

\subsection{Our results}

All of our positive results are based on a quantum subroutine due to
\cite{BSHJA}, which we will refer to as an $\bj$ (Fourier Sample)
oracle call. As explained in Section~\ref{sec:jprelim}, a call to
the $\bj$ oracle yields a subset of $\{1,\dots,n\}$ (this set should
be viewed as a subset of the input variables $x_1,\dots,x_n$ of $f$)
drawn according to the Fourier spectrum of the Boolean function $f$.
As demonstrated by \cite{BSHJA}, such an oracle can be implemented
using $O(1)$ uniform quantum examples from a uniform distribution
quantum example oracle.  In fact, all of our algorithms will be
purely classical apart from their use of the $\bj$ oracle. Thus, all
of our algorithms can be implemented within the (uniform
distribution) quantum PAC model first proposed by \cite{BSHJA}. This
model is a natural quantum extension of the classical PAC model
introduced by Valiant \cite{Val84}, as described in
Section~\ref{sec:jprelim}.   We emphasize that no membership
queries, classical or quantum, are used in our algorithms, only
uniform quantum superpositions of labeled examples, and we recall
that such uniform quantum examples cannot efficiently simulate even
classical membership queries in general (see \cite{BSHJA}).

Our approach of focusing only on the $\bj$ oracle allows us to
abstract away from the intricacies of quantum computation, and
renders our results useful in any setting in which an $\bj$ oracle
can be provided to the user. In fact, learning and testing with
$\bj$ oracle queries may be regarded as a new distinct model (which
may possibly be weaker than the uniform distribution quantum example
model).

 We are primarily interested in the information theoretic
requirements (i.e. the number of oracle calls needed) of the
learning and testing problems that we discuss.  We give upper and
lower bounds for a range of learning and testing problems
related to \emph{$k$-juntas}; these are Boolean functions $f: \bits^n
\rightarrow \bits$ that depend only on (an unknown subset of) at most $k$
of the $n$ input variables $x_1,\dots,x_n$.  Juntas have been the subject
of intensive research in learning theory and property
testing in recent years, see e.g.~\cite{AR, AR2, Blum, CG04,
FKRSS, LMMV, MOS04}.

Our first result, in Section~\ref{sec:testjuntas}, is a $k$-junta
testing algorithm which uses $O(k\epsilon^{-1})$ $\bj$ oracle calls.
Our algorithm uses fewer queries than the best known classical junta
testing algorithm due to Fischer {\em et al.} \cite{FKRSS}, which
uses $O((k\log k)^2 \epsilon^{-1})$ membership queries. However,
since the  best lower bound known for classical membership query
based junta testing (due to Chockler and Gutfreund \cite{CG04}) is
$\Omega(k)$, our result does not rule out the possibility that there
might exist a classical membership query algorithm with the same
query complexity.

To complement our $\bj$ based testing algorithm, we establish a new
lower bound: Any $k$-junta testing algorithm that uses only a $\bj$
oracle requires $\Omega(\sqrt{k})$ calls to the $\bj$ oracle.  This
shows that our testing algorithm is not too far from optimal.

Finally, we consider algorithms that can both make $\bj$ queries and
also access classical random examples.  In
Section~\ref{sec:learnjuntas} we give an algorithm for learning
$k$-juntas over $\bits^n$ that uses $O(\epsilon^{-1} k\log k)$ $\bj$
queries and $O(2^k\log(\eps^{-1}))$ random examples. Since any
classical learning algorithm requires $\Omega(2^{k}+\log n)$
examples (even if it is allowed to use membership queries), this
result illustrates that it is possible to reduce the classical query
complexity substantially (in particular, to eliminate the dependence
on $n$) if the learning algorithm is also permitted to have some
very limited quantum information. Moreover most of the consumption
of our algorithm is from classical random examples which are
considered quite ``cheap'' relative to quantum examples. From
another perspective, our result shows that for learning $k$-juntas,
almost all the quantum examples used by the algorithm of Bshouty and
Jackson \cite{BSHJA}  can in fact be converted into ordinary
classical random examples.  We show that our algorithm is close to
best possible by giving a nearly matching lower bound.

\subsection{Organization}
In Section~\ref{sec:jprelim} we describe the models and problems we
will consider and present some useful preliminaries from Fourier
analysis and probability. Section~\ref{sec:testjuntas} gives our
results on testing juntas and Section~\ref{sec:learnjuntas} gives
our results on learning juntas.

\section{Preliminaries}
\label{sec:jprelim}
\subsection{The problems and the models}
In keeping with standard terminology in learning theory, a
\emph{concept} $f$ over $\bits^{n}$ is a Boolean function $f:
\bits^{n} \to \bits$, where $-1$ stands for \textsc{True} and $1$
stands for \textsc{False}. A \emph{concept class} $\cc = \cup_{n
\geq 1} C_n$ is a set of concepts where $C_n$ consists of those
concepts in $\cc$ whose domain is $\bits^n.$ For ease of notation
throughout the paper we will omit the subscript in $C_n$ and simply
write $C$ to denote a collection of concepts over $\bits^n$.

The concept class we will chiefly be interested in
is the class of \emph{$k$-juntas}.  A Boolean function $\fisafunc$
is a $k$-junta if $f$ depends only on $k$ out of its $n$
input variables.

\subsubsection{The problems}
We are interested in the following computational problems:
\begin{description}
    \item[PAC Learning under the uniform distribution:]
    Given any \emph{target concept}
    $f \in C$, an {\em $\epsilon$-learning algorithm for concept class $C$} under the
    uniform distribution outputs a \emph{hypothesis} function $h: \bits^n \rightarrow \bits$
    which, with probability at least $2/3$, agrees with $c$ on at least a $1-\epsilon$
    fraction of the inputs in $\bits^n.$  This is a widely studied framework
    in the learning theory literature both in classical
    (see for instance \cite{KM,JACKSON}) and in quantum (see \cite{BSHJA}) versions.

    \item[Property testing:] Let $f$ be any Boolean function $f:
    \bits^{n} \rightarrow \bits$. A {\em property testing algorithm
    for concept class $C$} is an algorithm which, given access to $f$,
    behaves as follows:
    \begin{itemize}
        \item If $f \in C$ then the algorithm outputs \textsc{Accept}
        with probability at least $2/3$;
        \item If  $f$ is \emph{$\epsilon$-far} from any concept in $C$
        (i.e. for every concept $g\in C$, $f$ and $g$ differ on at least
        an $\epsilon$ fraction of all inputs),
        then the algorithm outputs \textsc{Reject} with probability at least $2/3$.
    \end{itemize}
    The notion of property testing was first developed by \cite{GGR} and \cite{RS96}. Quantum property testing was
    first studied by Buhrman {\em et al.} \cite{BFNR}, who first gave an example of an
    exponential separation between the query complexity of classical and quantum testers
    for a particular concept class.
\end{description}

Note that a learning or testing algorithm for $C$ ``knows'' the
class $C$ but does not know the identity of the concept $f$. While
our primary concern is the number of oracle calls that our
algorithms use, we are also interested in \emph{time efficient}
algorithms for testing and learning; for the concept class of
$k$-juntas, these are algorithms running in poly$(n,2^k,\eps^{-1})$
time steps.

\subsubsection{Classical oracles}
In order for learning and testing algorithms to gather information
about the unknown concept $f$, they need an information source
called an \emph{oracle}. The number of times an oracle is queried by
an algorithm is referred to as the \emph{query complexity}.
Sometimes our algorithms will be allowed access to more than one
type of oracle in our discussion.

In this paper we will consider the following types of oracles that
provide classical information:

\begin{description}
    \item[Membership oracle $\mq$:] For $f$ a Boolean function, a
\emph{membership oracle} $\mq(f)$ is an oracle which, when queried
with input $x$, outputs the label $f(x)$ assigned by $f$ to $x.$

\item[Uniform random example oracle $\ex$:] A query $\ex(f)$ of the random
example oracle returns an ordered pair
        $(x, f(x))$ where $x$ is drawn uniformly random from the set
        $\{-1,1\}^n$ of all possible inputs.
\end{description}

Clearly a single call to an $\mq$ oracle can simulate the random
example oracle $\ex$. Indeed $\ex$ oracle queries are considered
``cheap'' compared to membership queries. For example, in many
settings it is possible to obtain random labeled examples but
impossible to obtained the label of a particular desired example
(consider prediction problems dealing with phenomena such as weather
or financial markets).  We note that the set of concept classes that
are known to be efficiently PAC learnable from uniform random
examples only is rather limited, see e.g. \cite{KL,MANSOUR}. In
contrast,  there are known efficient algorithms that use membership
queries to learning important function classes such as $\dnf$
(Disjunctive Normal Form) formulas \cite{JACKSON}.

\subsubsection{Quantum oracles:}
We will consider the following quantum oracles, which are the
natural quantum generalizations of membership queries and uniform
random examples respectively.
\begin{description}
    \item[Quantum membership oracle $\qmq$:]
    The quantum membership oracle $\qmq(f)$ is the quantum oracle whose query acts on
        the computational basis states as follows:
        $$\qmq(f) \colon |x,b\rangle\mapsto|x,b\cdot f(x)\rangle,\ \text{where $x \in\bits^{n}$ and $b \in \bits$}.$$
    \item[Uniform quantum examples $\qex$:] The uniform quantum example oracle $\qex(f)$ is the quantum oracle whose query acts on the
        computational basis state $|1^{n},1\rangle$ as follows:
        $$\qex(f) \colon |1^{n},1\rangle \mapsto \sum_{x\in\bits^{n}}
        \frac{1}{2^{n/2}}|x,f(x)\rangle.$$
        The action of a $\qex(f)$ query is undefined on other basis states, and
        an algorithm may only invoke the $\qex(f)$ query on the basis state $|1^n,1\rangle$.
\end{description}

It is clear that a $\qmq$ oracle can simulate a $\qex$ oracle or an
$\mq$ oracle, and a $\qex$ oracle can simulate an $\ex$ oracle.

The model of PAC learning with a uniform quantum example oracle
 was introduced by Bshouty and Jackson in
\cite{BSHJA}.  Several researchers have also studied learning from a
more powerful $\qmq(f)$ oracle, see
e.g.~\cite{AKMPY,AS05,IKRY,RSSG}. Turning to property testing, we
are not aware of prior work on quantum testing using only the
$\qex(f)$ oracle; instead researchers have considered quantum
testing algorithms that use the more powerful $\qmq(f)$ oracle, see
e.g.~\cite{BFNR,FMSS,MN}.

\subsection{Harmonic analysis of functions over $\bits^{n}$}\label{sec:bfour}
We will make use of the Fourier expansion of real valued
functions over $\bits^{n}$.
We write
$[n]$ to denote the set of variables $\{x_1, x_2, \ldots, x_n\}$.

Consider the set of real valued functions over $\bits^{n}$ endowed with
the inner product $$\innp{f}{g} = \E[f g] = {\frac 1 {2^n}}
\sum_x f(x) g(x)$$
and induced norm
$\|f \|=\sqrt{\innp{f}{f}}$. For each $S\subseteq [n]$, let
$\chi_S$ be the parity function $\chi_S(x)= \prod_{x_i\in S} x_i.$
It is a well known fact that the $2^n$
functions $\{\chi_S(x), S \subseteq [n]\}$ form
an orthonormal basis for the vector space of real valued functions over $\bits^{n}$ with the above
inner product. Consequently, every $f \colon \bits^{n} \to \mathbb{R}$ can be expressed uniquely as:
\[ f(x)=\sum_{S\subseteq[n]}  \fousumeltbl\chi_S(x)
            \]
which we refer to as the \emph{Fourier expansion} or \emph{Fourier
transform} of $f$. Alternatively, the values $\{\fousumeltbl \colon
S \subseteq [n] \} $ are called the \emph{Fourier coefficients} or
the \emph{Fourier spectrum} of $f$.

\emph{Parseval's Identity}, which is an easy consequence of
orthonormality of the basis functions, relates the values of the
coefficients to the values of the function:
\begin{lemma}[Parseval's Identity]\label{bparseval} For any $f \colon \bits^{n} \to \mathbb{R}$, we have
$ \sum_{S\subseteq[n]} |\fousumeltbl|^2= \E[f^2]$. Thus for a Boolean valued function $ \sum_{S\subseteq[n]} |\fousumeltbl|^2=1$.
\end{lemma}
We will use the following simple and well-known
fact:
\begin{fact}[See \cite{KM}]\label{kmfact} For any $\fisafunc$
and any $g: \bits^n \rightarrow \R$, we have
    \[\P_{x}[f(x)\neq \sgn(g(x))]\leq\E_{x}[{(f(x)-g(x))}^{2}]=\sum_{S\subseteq[n]}|\fousumeltbl-\hat{g}(S)|^{2}\]
\end{fact}

Recall that the \emph{influence} of a variable $x_i$ on a Boolean
function $f$ is the probability (taken over a uniform random input $x$ for $f$)
that $f$ changes its value when the $i$-th bit of $x$ is flipped,
i.e.
\[\Inf_i(f) = \P_x[f(x_{i} \leftarrow -1) \neq f(x_{i} \leftarrow 1)].\]
It is well known (see e.g. \cite{KKL}) that
$\Inf_i(f) =\sum_{S \ni x_i} |\hat{f}(S)|^2.$

\subsection{Additional tools}
\begin{fact}[Data Processing Inequality]\label{dpi}Let $X_1, X_2$ be two random variables over the same domain. For any (possibly randomized)
    algorithm $\mathcal{A}$, one has that $$\|\mathcal{A}(X_1)-\mathcal{A}(X_2)\|_{1} \leq \|X_1-X_2\|_{1}.$$
\end{fact}
Let  $S_1, S_2$ be random variables corresponding to sequences of draws
taken from two different distributions over the same domain.
By the above inequality, if $\|S_1-S_2\|_{1}$ is known to be small, then
the probability of success must be small
for any algorithm designed to distinguish if the draws are made
according to $S_1$ or $S_2$.

We will also use standard Chernoff bounds on tails of sums of
independent random variables:

\begin{fact}[Additive Bound] Let
$X_1, \ldots, X_m$ be i.i.d. random variables with mean $\mu$ taking
values in the range $[a,b]$. Then
    for all $\lambda>0$ we have $\P[|\frac{1}{m}\sum_{i=1}^{m} X_i - \mu| \geq \lambda]\leq 2 \exp(\frac{-2\lambda^2 m}{(b-a)^2})$.
\end{fact}

\subsection{The Fourier sampling oracle: $\bj$}

\begin{definition} Let $f: \bits^n \rightarrow \bits$ be a Boolean function.
The \emph{Fourier sampling oracle} $\bj(f)$ is the classical oracle
which, at each invocation, returns each subset of variables $S
\subseteq \{1,\dots,n\}$ with probability $|\fousumeltbl|^2$, where
$\fousumeltbl$ denotes the Fourier coefficient corresponding to
$\chi_S(x)$ as defined in Section~\ref{sec:bfour}.
\end{definition}

This oracle will play an important role in our algorithms.  Note
that by Parseval's Identity we have $\sum_{S\subseteq[n]}
|\fousumeltbl|^2=1$ so the probability distribution over sets $S$
indeed has total weight 1.

In \cite{BSHJA} Bshouty and Jackson describe a simple constant-size
quantum network \texttt{QSAMP}, which has its roots in an idea from
\cite{BV97}. \texttt{QSAMP} allows sampling from the Fourier
spectrum of a Boolean function using $O(1)$ $\qex$ oracle queries:
\begin{fact}[See \cite{BSHJA}]\label{bjfact} For any Boolean function $f$,
it is possible to simulate a draw from the $\bj(f)$ oracle with
probability $1-\delta$ using $O(\log \delta^{-1})$ queries to
$\qex(f)$.
\end{fact}
All the algorithms we describe are actually classical algorithms
that make $\bj$ queries.

\section{Testing juntas}
\label{sec:testjuntas}

Fischer {\em et al.} \cite{FKRSS} studied the problem of testing
juntas given black-box access (i.e., classical membership query
access) to the unknown function $f$ using harmonic analysis and
probabilistic methods. They gave several different algorithms with
query complexity independent of $n$, the most efficient of which
yields the following:
\begin{theorem}[See {\cite[Theorem~6]{FKRSS}}]\label{S1Thm1}There is an algorithm
that tests whether an unknown $f: \bits^n \rightarrow \bits$ is a
$k$-junta using $O((k\log k)^2 \epsilon^{-1})$ membership queries.
\end{theorem}

Fischer \emph{et al.} also gave a lower bound on the number of queries
required for testing juntas, which was subsequently improved by
Chockler {\em et al.} to the following:
\begin{theorem}[See \cite{CG04}]\label{S1Thm2}
Any algorithm that tests whether $f$ is a $k$-junta or is $1/3$-far
from every $k$-junta must use $\Omega(k)$ membership queries.
\end{theorem}
We emphasize that that both of these results concern algorithms
with classical membership query access.

\subsection{A testing algorithm using
$O(k/\eps)$ $\bj$ oracle calls}
In this section we describe a new testing algorithm that uses the
$\bj$ oracle
and prove the following theorem about its performance:
\begin{theorem}\label{S1QThm1}There is an algorithm that tests the
property of being a $k$-junta using $O(k/\epsilon)$
calls to the $\bj$ oracle.
\end{theorem}
As described in Section~\ref{sec:jprelim}, the algorithm can thus
be implemented using $O(k/\eps)$ uniform quantum examples from $\qex(f)$.

\begin{proof}
Consider the following algorithm $\alg$  which has $\bj$ oracle access
to an unknown function $\fisafunc$.  Algorithm $\alg$ first makes
$10(k+1)/\eps$ calls to the $\bj$ oracle; let $\mathcal{S}$ denote the union of all the sets of variables received as responses to
these oracle calls.  Algorithm $\alg$ then outputs
``\textsc{Accept}'' if $|\mathcal{S}| \leq k$ and outputs ``\textsc{Reject}''
if $|\mathcal{S}|>k$.

It is clear that if $f$ is a $k$-junta then $\alg$ outputs ``\textsc{Accept}''
with probability 1. To prove correctness of the test
it suffices to show that if $f$ is
$\eps$-far from any $k$-junta then $\P[\alg$ outputs ``\textsc{Reject}''$] \geq {\frac 2 3}.$

The argument is similar to the standard analysis of the
coupon collector's problem.
Let us view the set $\mathcal{S}$ as growing incrementally step by step as successive calls to the $\bj$ oracle are performed.

Let $X_i$ be a random variable which denotes the number of $\bj$
queries that take place starting immediately after the $(i-1)$-st
new variable is added to $\mathcal{S}$, up through the draw when the
$i$-th new variable is added to $\mathcal{S}$.  If the $(i-1)$-st
and $i$-th new variables are obtained in the same draw then $X_i=0$.
(For example, if the first three queries to the $\bj$ oracle are
$\{1,2,4\},$ $\{2,4\}$, $\{1,4,5,6\}$, then we would have $X_1=1$,
$X_2 = 0$, $X_3=0$, $X_4 = 2$, $X_5=0$.)

Since $f$ is $\eps$-far from any $k$-junta, we know that for any set $\mathcal{T}$ of $k' \leq k$ variables, it must be the case that
\[
\sum_{S \subseteq \mathcal{T}} \hat{f}(S)^2 \leq 1 - \eps
\]
(since otherwise if we set $g=\sum_{S \subseteq \mathcal{T}} \hat{f}(S)
\chi_S, h=\sgn(g)$ and use Fact~\ref{kmfact}, we would have
\[
\P_{x}[f(x)\neq h(x)]\leq\E_{x}[{(f(x)-g(x))}^{2}]
= \sum_{S \not\subseteq \mathcal{T}} \hat{f}(S)^2 < \epsilon
\]
which contradicts the fact that $f$ is $\eps$-far from any $k$-junta).
It follows that for each $1 \leq i \leq k$, if at the current stage
of the construction of $\mathcal{S}$ we have $|\mathcal{S}|=i$, then the
probability that the
next $\bj$ query yields a new variable outside of $\mathcal{S}$ is at least
$\eps$.
Consequently we have $\E[X_i] \leq {\frac 1 \eps}$ for each $1 \leq i
\leq k+1$, and hence
\[
\E[X_1 + \cdots + X_{k+1}] \leq {\frac {(k+1)} \eps}.
\]
By Markov's inequality, the probability that $X_1 + \cdots + X_{k+1} \leq
10(k+1)/\eps$ is at least $9/10$, and therefore with probability at least
$9/10$ it will be the case after $10(k+1)/\eps$ draws that $|\mathcal{S}|>k$
and the algorithm will consequently output ``\textsc{Reject}.''
\end{proof}

Note that the
$O(k/\eps)$ uniform quantum examples required for Algorithm $\alg$ improves
on the $O((k \log k)^{2}/\eps)$ query complexity of the best
known classical algorithm.  However our result
does not conclusively show that $\qex$
queries are more powerful than classical membership queries for this problem
since it is conceivable that there could exist an as
yet undiscovered $O(k/\eps)$ classical membership query algorithm.

\subsection{Lower bounds for testing with a $\bj$ oracle}

\subsubsection{A first approach}

As a first attempt to obtain a lower bound on the number of
$\bj$ oracle calls required to test $k$-juntas, it is natural to consider
the approach of Chockler \emph{et al.} from \cite{CG04}.
To prove Theorem~\ref{S1Thm2}, Chockler \emph{et al.}
show that any classical algorithm which can successfully distinguish between
the following two probability distributions over black-box functions
must use $\Omega(k)$ queries:
\begin{itemize}
\item {\bf Scenario I:} The distribution ${\dist}^{(0)}_{k,n}$ is uniform
over the set of all Boolean functions over $n$ variables
which do not depend on variables $k+2,\dots,n.$
\item {\bf Scenario II:}
The distribution $\dist^{(1)}_{k,n}$ is defined as follows:
to draw a function $f$ from this distribution, first
an index $i$ is chosen uniformly from $1,\ldots,k+1$, and then $f$
is chosen uniformly from among those functions that do not depend on
variables $k+2,\dots,n$ or on variable $i$.
\end{itemize}

The following observation shows that this approach will not
yield a strong lower bound for algorithms that have
access to a $\bj$ oracle:
\begin{observation} With $O(\log k)$ queries to a
$\bj$ oracle, it is possible to determine w.h.p. whether
a function $f$ is drawn from Scenario I or Scenario II.
\end{observation}
\begin{proof}
It is easy to see that a function drawn from Scenario I is simply a
random function on the first $k+1$ variables.
The Fourier spectrum of random Boolean functions is studied
in \cite{OS03}, where it is shown that sums of squares of Fourier
coefficients of random Boolean functions
are tightly concentrated around their expected value.
In particular, Proposition~6 of \cite{OS03} directly implies that
for any fixed variable $x_i, i\in 1,\ldots,k+1,$ we have:
\[\P_{f\leftarrow {\dist}^{(0)}_{k,n}}\left[\sum_{S \ni x_i}
\hat{f}(S)^{2} > \frac{1}{3}\right] <\exp(-2^{k+1}/2592).
\]
Thus with overwhelmingly high probability, if $f$ is
drawn from Scenario I then each $\bj$ query will ``expose''
variable $i$ with probability at least $1/3$.   It
follows that after $O(\log k)$ queries all $k+1$ variables will have
been exposed; so by making $O(\log k)$ $\bj$ queries and simply
checking whether or not $k+1$ variables have been exposed, one
can determine w.h.p. whether $f$ is drawn from Scenario I or Scenario II.
\end{proof}
Thus we must adopt a more sophisticated approach to prove
a strong lower bound on $\bj$ oracle algorithms.

\subsubsection{An $\Omega(\sqrt{k})$ lower bound for
$\bj$ oracle algorithms}

Our main result in this section is the following theorem:
\begin{theorem}\label{S1QThm2}
Any algorithm that has $\bj$ oracle access to an unknown $f$ must use
$\Omega(\sqrt{k})$ oracle calls to test whether $f$ is a $k$-junta.
\end{theorem}

\begin{proof}
Let $k$ be such that $k=r+2^{r-1}$ for some positive integer $r.$ We
let $R$ denote $2^r.$ The \emph{addressing function} on $r+R$
variables has $r$ ``addressing variables,'' which we shall denote
$x_1,\dots,x_r,$ and $R=2^r$ ``addressee variables'' which we denote
$z_0,\dots,z_{R-1}.$ The output of the function is the value of
variable $z_{\mathbf{x}}$ where the ``address'' ${\mathbf{x}}$ is
the element of $\{0,\dots,R-1\}$ whose binary representation is
given by $x_1\ldots x_r$. Figure~1 depicts a decision tree that
computes the addressing function in the case $r=3$. Formally, the
Addressing function $\textsc{Addressing}: \{-1,1\}^{r+R} \rightarrow
\{-1,1\}$ is defined as follows:
\begin{align*}
&\text{\textsc{Addressing}}(x_1, x_2, \ldots, x_r, z_{0}, z_{1}, \ldots,
z_{R-1})=z_{\mathbf{x}}, \\
&\text{where}\
\mathbf{x}=(\frac{1-x_{1}}{2})\circ
(\frac{1-x_{2}}{2})\circ\ldots \circ(\frac{1-x_{r}}{2})\ \text{in
binary form and $\circ$ is binary concatenation}.
\end{align*}

\begin{figure}[t]
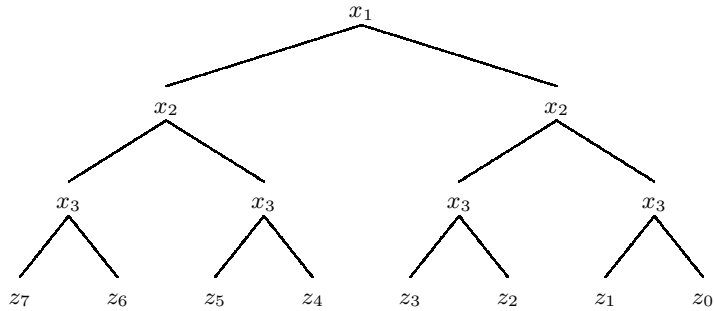

\begin{center}
\label{fig:tree}
\childsidesep{1cm}
\childattachsep{1cm}
\synttree{3}[$x_1$[$x_2$ [$x_3$ [$z_7$] [$z_6$]] [$x_3$ [$z_5$] [$z_4$]]] [$x_2$ [$x_3$ [$z_3$] [$z_2$]] [$x_3$ [$z_1$ ] [$z_0$]]]]
 \caption{A decision tree computing the addressing function in
the case $r=3$.  The left edge out of each node corresponds to the
variable at the node taking value $-1$ and the right edge to the
variable taking value 1.}
\end{center}
\end{figure}

Intuitively, the Addressing function will be useful for us because
as we will see the Fourier spectrum is ``spread out'' over the $R$
addressee variables; this will make it difficult to distinguish the
Addressing function (which is not a $k$-junta since $k=r+R/2$ and as
we shall see is in fact far from every $k$-junta) from a variant
which is a $k$-junta.

Let $x_1,\dots,x_r,y_0,\dots,y_{n-r-1}$ be the $n$
variables that our Boolean functions are defined over.
We now define two distributions
$\dist_{\textsc{Reject}}$,
$\dist_{\textsc{Accept}}$
over functions on these variables.

 The distribution $\dist_{\textsc{Reject}}$ is
defined as follows:  to make a draw from $\dist_{\textsc{Reject}}$,

\begin{enumerate}
\item First uniformly choose a subset $T$ of $R$ variables from
$\{y_0,\dots,y_{n-r-1}\}$;
\item Next, replace the variables
$z_{0},\dots,z_{R-1}$ in the function
\[\textsc{Addressing}(x_1, \ldots, x_r, z_{0}, \ldots,
z_{R-1})\] with the variables in $T$ (choosing the variables from $T$
in a uniformly random order).  Return the resulting function.
\end{enumerate}
Note that step (2) in the description of
making a draw from $\dist_{\textsc{Reject}}$ above corresponds to
placing the variables in $T$ uniformly at the leaves of the decision
tree for $\textsc{Addressing}$ (see Figure~1).

Equivalently, if we write $f_\tau$ to denote the following function \emph{over $n$ variables}
\begin{equation}\label{eqn:ftau1}
f_{\tau}(x_1,\ldots, x_r, y_{0}, \ldots, y_{n-r-1})=
\textsc{Addressing}(x_1, x_2, \ldots, x_r, y_{\tau(0)},
y_{\tau(1)}, \ldots, y_{\tau(R-1)});
\end{equation}
a draw from
$\dist_{\textsc{Reject}}$ is a function chosen uniformly at random
from the set
$C_{\textsc{Reject}}= \{f_{\tau}\}
$
where $\tau$ ranges over all permutations of
$\{0,\dots,n-r-1\}.$

It is clear that every function in $C_{\textsc{Reject}}$ (the support of $\dist_{\textsc{Reject}}$)
depends on $r+R$ variables and thus is not a $k$-junta.  In fact,
every function in $C_{\textsc{Reject}}$ is far from being a $k$-junta:

\begin{lemma}
Every $f$ that has nonzero probability under $\dist_{\textsc{Reject}}$ is
$1/6$-far from any $k$-junta.
\end{lemma}
\begin{proof} Fix any such $f$ and let $g$ be any $k$-junta.  It is clear
that at least $R/2 - r$ of the ``addressee'' variables of $f$ are not
relevant variables for $g$.
For a ${\frac {R/2 - r}{R}} > 1/3$ fraction of all inputs to $f$,
the value of $f$ is determined by one of these addressee variables;
on such inputs the error rate of $g$ relative to $f$ will be precisely
$1/2.$
\end{proof}

Fix any function $f_\tau$ in $C_{\textsc{Reject}}$.  We now give an
expression for the Fourier representation of $f_\tau$.  The
expression is obtained by viewing $f_\tau$ as a sum of $R$
subfunctions, one for each leaf of the decision tree, where each
subfunction takes the appropriate nonzero value on inputs which
reach the corresponding leaf and takes value 0 on all other
inputs:
\small
\begin{align}\label{eqn:ftau2}
    f_{\tau}(x_1,\ldots, x_r, y_{0}, \ldots, y_{n-r-1})& = \sum_{\mathbf{i}=i_{1}i_{2}\ldots i_{r}=0}^{R-1} y_{\tau(\mathbf{i})}
(\frac{1 + (-1)^{i_{1}}x_{1}}{2})(\frac{1 + (-1)^{i_{2}}x_{2}}{2})\ldots(\frac{1 + (-1)^{i_{r}}x_{r}}{2})
\end{align}
\begin{align}
\quad & = \frac{1}{2^{r}}\sum_{\mathbf{i}=0}^{R-1}
\sum_{X\subseteq\{x_1,\ldots,x_r\}}(-1)^{(\sum_{x_{j}\in X} i_{j})}
y_{\tau(\mathbf{i})} \chi_{X}.\label{eqn:ftau2l2}
\end{align}
\normalsize Note that whenever $\frac{1-x_1}{2}=i_1,
\frac{1-x_2}{2}=i_2, \ldots, \frac{1-x_r}{2}=i_r$, the sum on the
RHS of Equation~\eqref{eqn:ftau2} has precisely one non-zero term
which is $y_{\tau(\mathbf{i})}$. This is because the rest of the
terms are annihilated since in each of these terms there is some
index $j$ such that $\frac{1-x_j}{2}=1- i_j$ which makes $(\frac{1 +
(-1)^{i_{j}}x_{j}}{2})=0$. Consequently this sum gives rise to
exactly the Addressing function in Equation~\eqref{eqn:ftau1} which
is defined as $f_{\tau}$ and consequently the equality in
Equation~\eqref{eqn:ftau2} follows. Equation~\eqref{eqn:ftau2l2}
follows easily from rearranging~\eqref{eqn:ftau2}.

Now we turn to $\dist_{\textsc{Accept}}.$

The distribution
$\dist_{\textsc{Accept}}$ is defined as follows: to make a draw from
$\dist_{\textsc{Accept}}$,

\begin{enumerate}
\item First uniformly choose a subset $T$ of $R/2$ variables from
$\{y_0,\dots,y_{n-r-1}\}$;
\item Next, replace the variables
$z_{0},\dots,z_{R/2-1}$ in the function
\[\textsc{Addressing}(x_1, \ldots, x_r, z_{0}, \ldots,
z_{R-1})\] with the variables in $T$ (choosing the variables from $T$
in a uniformly random order).
\item Finally, for each $\mathbf{i}=0,\dots,R/2-1$ do the following:
    if variable $y_j$ was used to replace variable $z_{\mathbf{i}}$ in the previous
step, let $s_{\mathbf{i}}$ be a fresh uniform random $\pm 1$ value and replace
variable $z_{R-1-\mathbf{i}}$ with $s_{\mathbf{i}} y_j$. Return the resulting function.
\end{enumerate}
Observe that for any integer $0 \leq \mathbf{i} < R/2$ with binary
expansion $\mathbf{i} = i_1  \circ i_2 \circ
 \cdots  \circ i_r$, we have that the binary expansion of
$R-1-\mathbf{i}$ is $\conj{i_1} \circ \conj{i_2} \circ \cdots \circ
\conj{i_r}$.  Thus steps (2) and (3) in the description of making a
draw from $\dist_{\textsc{Accept}}$ may be restated as follows in
terms of the decision tree representation for $\textsc{Addressing}$:
\begin{itemize}
\item [$2'.$] Place the variables $y_j\in T$ randomly among the leaves of the
decision tree with index less than $R/2$.
\item [$3'.$] For each variable $y_j\in T$ placed at the leaf with index
$\mathbf{i}=i_1 \circ  i_2 \circ \cdots \circ i_r<R/2$ above, throw
a $\pm 1$ valued coin $s_{\mathbf{i}}$ and place $s_{\mathbf{i}}
y_j$ at the antipodal leaf location with index:
            $\conj{\mathbf{i}}=\conj{i_1} \circ \conj{i_2} \circ \cdots \circ \conj{i_r}=R-1-\mathbf{i}$.
\end{itemize}

Equivalently, if we write $g_{\tau,s}$ to denote the following function \emph{over $n$ variables}
\begin{align}
&g_{\tau,s}(x_1,\ldots, x_r, y_{0}, \ldots, y_{n-r-1})=\notag\\
&\textsc{Addressing}(x_1, \ldots, x_r, y_{\tau(0)},
\ldots,y_{\tau(R/2-1)}, s_{(R/2-1)}y_{\tau(R/2-1)},
\ldots  ,s_{0}y_{\tau(0)});\label{eqn:gtaus}
\end{align}
a draw from
$\dist_{\textsc{Accept}}$ is a function chosen uniformly at random
from the set
$C_{\textsc{Accept}}= \{g_{\tau,s}\}$
where $\tau$ ranges over all permutations of $\{0,\dots,n-r-1\}$ and
$s$ ranges over all of $\{-1,1\}^{R/2}$. It is clear that every function in
$C_{\textsc{Accept}}$ depends on at most $r+R/2=k$ variables, and
thus is indeed a $k$-junta.

By considering the contribution
to the Fourier spectrum from each pair of leaves
$\mathbf{i},\conj{\mathbf{i}}$ of the decision tree, we obtain the
following expression for the Fourier expansion of each function in
the support of $\dist_{\textsc{Accept}}$:
\small
\begin{align}
    g_{\tau,\mathbf{s}}(x_1, \ldots, x_r, y_{0}, \ldots, y_{n-r-1})=&\sum_{\mathbf{i}=i_{1}i_{2}\ldots i_{r}=0}^{R/2-1} y_{\tau(\mathbf{i})}
(\frac{1 + (-1)^{i_{1}}x_{1}}{2})(\frac{1 + (-1)^{i_{2}}x_{2}}{2})\ldots(\frac{1 + (-1)^{i_{r}}x_{r}}{2})\notag\\
 +&\sum_{\mathbf{i}=0}^{R/2-1} s_{\mathbf{i}} y_{\tau(\mathbf{i})}
(\frac{1 + (-1)^{\conj{i_{1}}}x_{1}}{2})(\frac{1 + (-1)^{\conj{i_{2}}}x_{2}}{2})\ldots(\frac{1 + (-1)^{\conj{i_{r}}}x_{r}}{2})\label{eqn:gtaus2}
\end{align}
\begin{equation}\label{eqn:gtaus2l2}
\text{[Since $(-1)^{\conj{i_{j}}}= -(-1)^{i_{j}}$]}\quad =\frac{1}{2^{r-1}}\sum_{\mathbf{i}=0}^{R/2-1}\begin{cases}\displaystyle
    \sum_{X\subseteq\{x_1,\ldots,x_r\}, |X|\ \text{even}} (-1)^{(\sum_{x_{j}\in X} i_{j})} y_{\tau(\mathbf{i})} \chi_{X} & \text{if $s_{\mathbf{i}}=1$;}\\
     \displaystyle \sum_{X\subseteq\{x_1,\ldots,x_r\}, |X|\ \text{odd}}  (-1)^{(\sum_{ x_{j}\in X} i_{j})} y_{\tau(\mathbf{i})} \chi_{X} & \text{if $s_{\mathbf{i}}=-1$.}\\
\end{cases}
\end{equation}
\normalsize
Just as in the Equation~\eqref{eqn:ftau2}, whenever
$\frac{1-x_1}{2}=i_1, \frac{1-x_2}{2}=i_2, \ldots, \frac{1-x_r}{2}=i_r$, the sum on the
RHS of Equation~\eqref{eqn:gtaus2} has precisely one non-zero term which is
$y_{\tau(\mathbf{i})}$ if $\mathbf{i}<R/2$ and $s_{R-1-\mathbf{i}} y_{\tau(R-1-\mathbf{i})}$ if $\mathbf{i}\geq R/2$.
Therefore this sum gives rise to exactly the Addressing function in Equation~\eqref{eqn:gtaus} which
is defined as $g_{\tau,s}$ and consequently the equality in Equation~\eqref{eqn:gtaus2} follows.

It follows that for each $g_{\tau,\mathbf{s}}$ in the support of
$\dist_{\textsc{Accept}}$ and for any fixed $y_j$, all elements of
the set $\{S \colon y_j \in S\ \text{and}\
\widehat{g_{\tau,\mathbf{s}}}(S)\neq 0\}$ will have the same parity.
Moreover, when draws from $\dist_{\textsc{Accept}}$ are considered, for
every distinct $y_j$ this odd/even parity is independent and
uniformly random.

Now we are ready to prove Theorem~\ref{S1QThm2}. Recall that a $\bj$
oracle query returns $S$ with probability $|\hat{f}(S)|^{2}$ for
every subset $S$ of input variables to the function. Considering the equations
\eqref{eqn:ftau2l2} and \eqref{eqn:gtaus2l2}, for any $f$ in $C_{\textsc{Accept}}$ or
$C_{\textsc{Reject}}$ its $\bj$ oracle will return a pair of the form $(y_{j=\tau(\mathbf{i})},X),\ X\subseteq\{x_1,\ldots,x_r\}$.

Let us define a set ${\cal T}$ of ``typical'' outcomes from $\bj$
oracle queries.  Fix any $N=o(\sqrt{k})$, and let ${\cal T}$ denote
the set of all sequences $\{(y_{j_1},X_1),\ldots,(y_{j_N},X_N)\}$ of
length $N$ which have the property that {\em no $y_{i}$ occurs
more than once among $y_{j_1},\dots,y_{j_N}$}.

Note that for any fixed $f_{\tau}\leftarrow \dist_{\textsc{Reject}}$,
every non-zero Fourier coefficient $\widehat{f_{\tau}}(S)$ satisfies $|\widehat{f_{\tau}}(S)|^{2}=\frac{1}{2^{2r}}=\frac{1}{R^2}$
due to Equation~\eqref{eqn:ftau2l2}. Therefore after $f_{\tau}$ is drawn, for any fixed $y_j$
the probability of receiving a response of the form $(y_j, X)$ as the outcome of a $\bj$
query is either
\begin{description}
    \item[$=0$,]if $f_{\tau}$ is not a function of $y_j$,
i.e. $j\notin \{\tau(0),\ldots,\tau(R-1)\}$; or
\item[$=\frac{1}{R}$,]if $j\in \{\tau(0),\ldots,\tau(R-1)\}$. This is because each of the $2^{r}=R$ responses $(y_j, X)$
    occurs with probability $\frac{1}{R^2}$.
    \end{description}

Similarly, for any fixed $g_{\tau,\mathbf{s}}\leftarrow\dist_{\textsc{Accept}}$,
every non-zero Fourier coefficient $\widehat{g_{\tau,\mathbf{s}}}(S)$ satisfies $|\widehat{g_{\tau,\mathbf{s}}}(S)|^{2}=\frac{1}{2^{2r-2}}=\frac{4}{R^2}$
due to Equation~\eqref{eqn:gtaus2l2}. Therefore after $g_{\tau,\mathbf{s}}$ is drawn, for any fixed $y_j$
the probability of receiving a response of the form $(y_j, X)$ as the outcome of a $\bj$ query is either
\begin{description}
\item[$=0$,]if $g_{\tau,\mathbf{s}}$ is not a function of $y_j$,
i.e. $j\notin \{\tau(0),\ldots,\tau(R/2-1)\}$; or
\item[$=\frac{2}{R}$,]if $j\in \{\tau(0),\ldots,\tau(R/2-1)\}$. This is because each of the $2^{r-1}=R/2$ responses $(y_j, X)$
    occurs with probability $\frac{4}{R^2}$.
\end{description}

Now let us consider the probability of obtaining
a sequence from ${\cal T}$ under each scenario.
\begin{itemize}
\item If the function is drawn from $\dist_{\textsc{Reject}}$:  the
probability is at least \[1(1-1/R)(1-2/R)\ldots(1-N/R) > 1-o(1)\quad
\text{[by the Birthday Paradox]}.\]
\item If the function is from $\dist_{\textsc{Accept}}$:  the
probability is at least
\[1(1-2/R)(1-4/R)\ldots(1-2N/R) > 1-o(1) \quad \text{[by the Birthday
Paradox]}\]
\end{itemize}

Now the crucial observation is that whether the function is drawn from  $\dist_{\textsc{Reject}}$ or from $\dist_{\textsc{Accept}}$,
each sequence in ${\cal T}$ is equiprobable by symmetry in the construction. To see this,
simply consider the probability of receiving a fixed $(y_{j}, X)$ for some new $y_{j}$ in the next $\bj$ query
of an unknown function drawn from either one of these distributions. Using the above calculations
for $|\hat{f}(y_{j}, X)|^{2}$, one can directly calculate that these probabilities are equal in either scenario.
Alternatively, for a function drawn from $\dist_{\textsc{Accept}}$ one can observe that since each successive $y_{j}$
is ``new'', a fresh random bit determines whether the support is an $(y_{j}, X)$ with $|X|$ odd or even;
once this is determined, the choice of $X$ is uniform from all subsets with the correct parity. Thus
the overall draw of $(y_j, X)$ is uniform over all $X$'s. Considering that the subset of relevant variables $T, |T|=R/2$ is
uniformly chosen from $\{y_{0}, \ldots, y_{n-r-1}\}$, this gives the equality of the probabilities
for each $(y_j, X)$ with a new $y_{j}$ when the function is drawn from $\dist_{\textsc{Accept}}$. The argument
for the case of $\dist_{\textsc{Reject}}$ is clear.

Consequently the statistical difference between the distributions
corresponding to the sequence of outcomes of the $N$ $\bj$ oracle
calls under the two distributions is at most $o(1)$. Now
Fact~\ref{dpi} implies that no algorithm making only $N$ oracle
calls can distinguish between these two scenarios with high
probability.  This gives us the result, and concludes the
proof of Theorem~\ref{S1QThm2}.
\end{proof}

Intuitively, under either distribution on functions, each element of
a sequence of $N$ $\bj$ oracle calls will ``look like'' a uniform
random draw $X$ from subsets of $\{x_1,\ldots,x_r\}$ and $j$ from
$\{0,\ldots,n-r-1\}$ where $j$ and $X$ are independent. Note that
this argument breaks down at $N=\Theta(\sqrt{R})$. This is because
if the algorithm queried the $\bj$ oracle $\Theta(\sqrt{R})$ times
it will start to see some $y_i$'s more than once with constant probability
(again by the birthday paradox). But when the functions are drawn from
$\dist_{\textsc{Accept}}$ the corresponding $X_i$'s will always have
a fixed parity for a given $y_i$ whereas for functions drawn from
$\dist_{\textsc{Reject}}$ the parity will be random each time. This
will provide the algorithm with sufficient evidence to distinguish
with constant probability between these two scenarios.

\section{Learning juntas}
\label{sec:learnjuntas}
\subsection{Known results}

The problem of learning an unknown $k$-junta
has been well studied in the computational
learning theory literature, see e.g. \cite{MOS04,AR,Blum}.
The following classical lower bound will be a yardstick against
which we will measure our results.
\begin{lemma}\label{S2Lem1}
Any classical membership query algorithm for learning $k$-juntas
to accuracy $1/5$ must use
$\Omega(2^{k}+\log n)$ membership queries.
\end{lemma}
\begin{proof}
Consider the restricted problem of learning an unknown function
$f(x)$ which is simply a single Boolean variable from $\{x_1,
\dots,x_n\}$. Since any two variables disagree on half of all
inputs, any $1/5$-learning algorithm can be easily modified into an
algorithm that exactly learns an unknown variable with no more
queries.  It is well known that any set of $n$ concepts requires
$\Omega(\log n)$ queries for any exact learning algorithm that uses
membership queries only, see e.g. \cite{BCG+96}. This gives the
$\Omega(\log n)$ lower bound.

For the $\Omega(2^k)$ lower bound, we may suppose that
the algorithm ``knows'' that the junta has relevant variables
$x_1,\dots,x_k$.  Even in this case, if fewer than ${\frac 1 2}2^k$
membership queries are made the learner will have no information
about at least $1/2$ of the function's output values.
A straightforward application of the Chernoff bound shows that it is
very unlikely for such a learner's hypothesis to be $1/5$-accurate,
if the target junta is a uniform random function over the relevant
variables.  This establishes the result.
\end{proof}

Learning juntas from uniform random examples $\ex(f)$ is a seemingly
difficult computational problem.
Simple algorithms based on exhaustive search can learn from
$O(2^k \log n)$ examples but require $\Omega(n^k)$ runtime.
The fastest known algorithm in this setting,  due to
Mossel {\em et al.}, uses $(n^k)^{{\frac \omega {\omega + 1}}}$ examples
and runs in $(n^k)^{{\frac \omega {\omega + 1}}}$ examples
time,
where $\omega < 2.376$ is the matrix multiplication exponent \cite{MOS04}.

Bshouty and Jackson \cite{BSHJA} gave an algorithm using uniform quantum
examples from the $\qex$ oracle to learn general $\dnf$ formulas.
Their algorithm uses $\tilde{O}(n s^{6} \eps^{-8})$ calls to $\qex$ to learn an $s$-term
$\dnf$ over $n$ variables to accuracy $\eps$.  Since any $k$-junta is expressible
as a $\dnf$ with at most $2^{k-1}$ terms, their result immediately
yields the following statement.
\begin{theorem}[See \cite{BSHJA}]\label{S2Thm2} There exists an $\epsilon$-learning quantum algorithm
    for $k$-juntas using $\tilde{O}(n 2^{6k} \epsilon^{-8})$ quantum examples under the uniform
    distribution quantum PAC model.
\end{theorem}
Note that \cite{BSHJA} did not try to optimize the quantum query complexity of their algorithms
in the special case of learning juntas. In contrast, our goal is to obtain a more efficient algorithm for juntas.

The lower bound of \cite[Observation~6.3]{AS05} for learning with quantum membership queries for an
arbitrary concept class
can be rephrased for the purpose of learning $k$-juntas as follows.
\begin{fact}[See \cite{AS05}] \label{S2Fac1}
Any algorithm for learning $k$-juntas to accuracy
$\eps = 1/10$ with quantum membership queries must use
$\Omega(2^k)$ queries.
\end{fact}
\begin{proof} Since we are proving a lower bound we may assume
that the algorithm is told in advance that the junta depends
on variables $x_1,\dots,x_k.$
Consequently we may assume that the algorithm makes all its queries
with nonzero amplitude only on inputs of
the form $|x, 1^{n-k} \rangle$.  Now \cite[Observation~6.3]{AS05} states
that any quantum algorithm which makes queries only over
a shattered set (as is the set of inputs $\{|x,1^{n-k} \rangle\}_{x \in \bits^k}$ for the class of $k$-juntas) must make at least VC-DIM($C$)/100 $\qmq$
queries to learn with error rate at most $\epsilon= 1/10$; here
VC-DIM($C$) is the Vapnik-Chervonenkis dimension of concept class $C$.
Since the VC dimension of the class of all Boolean functions
over variables $x_1,\dots,x_k$ is $2^k$,
the result follows.
\end{proof}
This shows that  a $\qmq$ oracle cannot provide sufficient information to learn a $k$-junta using $o(2^k)$ queries to high accuracy. It is worth noting that there are other similar learning problems known where an $N$-query $\qmq$ algorithm can
exactly identify a target concept whose description length is $\omega(N)$ bits. For instance, a single $\bj$ oracle call (which can be implemented by a
single $\qmq$ query) can potentially give up to $k$ bits of information;
if the concept class $C$ is the class of all $2^k$ parity functions
over the first $k$ variables, then any concept in the class can be
exactly learned by a single $\bj$ oracle call.

Note that all the results we have discussed in this subsection concern
algorithms with access to only one type of oracle; this is in contrast
with the algorithm we present in the next section.
\subsection{A new learning algorithm}

The motivating question for this section is: ``Is it possible to reduce
the classical query/sample complexity drastically for the problem of
junta learning if the learning algorithm is also permitted to have
very limited quantum information?''
We will give an affirmative answer to this question by describing a new
algorithm that uses both $\bj$ queries (i.e. quantum examples) and
classical uniform random examples.

\begin{lemma}\label{S1QLem1}
    Let $\fisafunc$ be a function whose value depends on the set of variables $\mathcal{I}$. Then there is
    an algorithm querying the $\bj$ oracle $O(\epsilon^{-1}\log |\mathcal{I}|)$ times which w.h.p. outputs a list
    of variables such that
    \begin{itemize}
        \item the list contains all the variables $x_i$ for which $\Inf_i(f) \geq \eps$; and
        \item all the variables $x_j$ in the list have non-zero influence: $\Inf_j(f) > 0$.
    \end{itemize}
\end{lemma}
\begin{proof} The algorithm simply queries the $\bj$ oracle  $N = O(\eps^{-1} \log |{\cal I}|)$ many times and outputs the union of all the sets of variables received
as responses to these queries.

If $\Inf_i(f) \geq \eps$ then the probability that $x_i$ never occurs in any
response obtained from the $N$ $\bj$ oracle calls is at most $(1 - \eps)^{N}
\leq {\frac 1 {10 |{\cal I}|}}.$  The union bound now yields that with probability at least $9/10$, every $x_i$ with $\Inf_i(f) \geq \eps$ is output by the algorithm.
\end{proof}

\begin{theorem}\label{S2Thm3} There is an efficient algorithm $\epsilon$-learning $k$-juntas with $O(\epsilon^{-1} k\log k)$ queries of the
    $\bj$ oracle and $O(2^k \log(\eps^{-1}))$ random examples.
\end{theorem}
\begin{proof}  We claim Algorithm~1 satisfies these requirements.

\begin{algorithm}[h]
        \caption{The junta learning algorithm.}
    \begin{algorithmic}[1]
        \STATE \textbf{Input:} $\epsilon>0, \bj(f), \ex(f)$.
        \STATE \textbf{Stage 1:}
        \STATE Construct a set containing all variables of $f$ with an influence at least
        $(\epsilon/10k)$ using the algorithm in Lemma~\ref{S1QLem1}.
        Let $\mathcal{A}$ be the final result.
        \STATE $\forall \mathbf{a}\in\bits^{|\mathcal{A}|}, encountered(\mathbf{a})\leftarrow \textsc{False}$.
        \STATE \textbf{Stage 2:}
        \REPEAT
        \STATE $\langle x, f(x) \rangle \leftarrow$ Draw from $\ex(f)$. Let $x|_{\mathcal{A}}$ denote
        the projection of $x$ onto the variables in $\mathcal{A}$.
        \IF{$encountered(x|_{\mathcal{A}})=\textsc{False}$}
        \STATE $value(x|_{\mathcal{A}})\leftarrow f(x), encountered(x|_{\mathcal{A}}) \leftarrow \textsc{True}$.
        \ENDIF
        \UNTIL{ $\text{For at least $(1-\epsilon/3)$ fraction of all}\ \mathbf{a}\in\bits^{|\mathcal{A}|}, encountered(\mathbf{a})=\textsc{True}$}.
        \STATE Output the hypothesis:
        \[H(x)=\begin{cases}
            value(x|_{\mathcal{A}}) & \text{if}\ encountered(x|_{\mathcal{A}})= \textsc{True}\\
            \text{\textsc{True}} & \text{otherwise}.\end{cases}\]
    \end{algorithmic}
\end{algorithm}

    Assume we are given a Boolean function $f$ whose value depends on the set of variables $\mathcal{I}$ with
    $|\mathcal{I}|\leq k$. By Lemma~\ref{S1QLem1}, $O(\epsilon^{-1} k\log k)$ queries of the $\bj$ oracle will reveal
    all variables with influence at least $(\epsilon/10k)$ with high probability during Stage 1.

    Assuming the algorithm of Lemma~\ref{S1QLem1} was successful, we group the variables as follows:

\bigskip

    \begin{tabular}{c c}
        Group  & Description \\ \hline
        $\mathcal{A}$ & The set of variables encountered in Stage 1.\\
        $\mathcal{B}$ & The set of relevant variables $\mathcal{I} \setminus \mathcal{A}$.\\
        $\mathcal{C}$ & The remaining $n-|\mathcal{I}|$ variables the function does not depend on.\\
    \end{tabular}

\bigskip

    Note that $|\mathcal{A}|+|\mathcal{B}|\leq k$ by Lemma~\ref{S1QLem1} and by the assumption that $f$ is a $k$-junta.
    \medskip

    We reorder the variables of $f$ so that the new order is $\mathcal{A}, \mathcal{B}, \mathcal{C}$
    for notational simplicity, i.e. $f$ is now considered to be over
    $(a_1,\ldots,a_{|\mathcal{A}|}, b_1, \ldots, b_{|\mathcal{B}|}, c_1, \ldots, c_{|\mathcal{C}|})$.
    We will denote an assignment to these variables by $(\mathbf{a},\mathbf{b},\mathbf{c})$.

    In Stage 2 the algorithm draws random examples until at least $(1 - \eps/3)$ fraction of all assignments
    to the variables in $\mathcal{A}$ are observed. Let us call this set of assignments by $\mathcal{S}$,
    and for every $\mathbf{a}\in \mathcal{S}$, let us denote the first example $\langle x,f(x)\rangle$ drawn in Stage 2
    for which $x|_{\mathcal{A}}=\mathbf{a}$ by $x=(\mathbf{a}, \mathbf{b}^{\mathbf{a}}, \mathbf{c}^{\mathbf{a}})$.
    At the end of the algorithm, the following hypothesis is produced as the output:
    $$H(\mathbf{a},\ast,\ast)=\begin{cases}
        f(\mathbf{a}, \mathbf{b}^{\mathbf{a}}, \mathbf{c}^{\mathbf{a}}) & \text{if}\ \mathbf{a}\in \mathcal{S}\\
        \textsc{True} & \text{otherwise}.\end{cases}$$
    In other words, the value of the hypothesis only depends on the setting of the variables in $\mathcal{A}$.
    Observe the probability that any given setting of a fixed set of variables in $\mathcal{A}$ has not been seen
    can be made less than $\eps/50$ using $O(\log(\eps^{-1}) 2^k)$ uniform random examples. Therefore the linearity
    of expectation implies that after $O(\log(\eps^{-1}) 2^k)$ random examples, the expected fraction of unseen
    assignments is $<\eps/50$. Thus by Markov's Inequality the fraction of unseen assignments will be $\leq\eps/3$
    w.h.p. Hence Stage 2 will terminate w.h.p. after $O(\log(\eps^{-1}) 2^k)$
    random examples. Consequently, the whole algorithm terminates with high probability with the desired query consumption.
    All we need to verify is that the hypothesis constructed is $\epsilon$-accurate.

    \textbf{The hypothesis $H$ is $\epsilon$-accurate with high probability:}

    We introduce some notation: Let $\mathbb{B}=\bits$; and given two strings
    $u, v\in\mathbb{B}^{\ell}$, let $u \odot v$ denote the bitwise multiplication between $u, v$; and let $|u|$ denote the
    total number of $-1$'s in $u$. Also let $\mathbf{1}_{W}$ denote the indicator function that takes value $1$ if
    $W$ holds and value $0$ if $W$ is false.

    We start with the following fact:
    \begin{fact}\label{fac:stupidfact}
        For any $s\in\mathbb{B}^{|\mathcal{B}|}$, we have $\frac{1}{2^{n}} \displaystyle\sum_{\mathbf{a}\in\mathbb{B}^{|\mathcal{A}|}}
        \sum_{\mathbf{b}\in\mathbb{B}^{|\mathcal{B}|}} \sum_{\mathbf{c}\in\mathbb{B}^{|\mathcal{C}|}}\mathbf{1}_{[f(\mathbf{a},\mathbf{b}\odot s,\mathbf{c})\neq f(\mathbf{a},\mathbf{b},\mathbf{c})]}< \eps/10$.
    \end{fact}
    \begin{proof}
    Given any string $ s\in \mathbb{B}^{|\mathcal{B}|}$, clearly there exists a sequence of $| s|+1$
    strings:
    $$1^{|\mathcal{B}|}= u^{1},  u^{2}, \ldots,  u^{|s|+1}= s,\ \text{where}\ u^{i}\in\mathbb{B}^{|\mathcal{B}|},\ \text{and for}\ i=1,\ldots,s,\ |u^{i}\odot u^    {i+1}|=1.$$
    Therefore,
    \begin{align*}
        \text{For any}\ s\in\mathbb{B}^{|\mathcal{B}|},&\quad \frac{1}{2^{n}} \displaystyle\sum_{\mathbf{a}\in\mathbb{B}^{|\mathcal{A}|}}
        \sum_{\mathbf{b}\in\mathbb{B}^{|\mathcal{B}|}} \sum_{\mathbf{c}\in\mathbb{B}^{|\mathcal{C}|}}\mathbf{1}_{[f(\mathbf{a},\mathbf{b}\odot s,\mathbf{c})\neq f(\mathbf{a},\mathbf{b},\mathbf{c})]} \\
        &\leq\frac{1}{2^{n}} \displaystyle\sum_{\mathbf{a}\in\mathbb{B}^{|\mathcal{A}|}}
        \sum_{\mathbf{b}\in\mathbb{B}^{|\mathcal{B}|}} \sum_{\mathbf{c}\in\mathbb{B}^{|\mathcal{C}|}}\sum_{i=1}^{|s|}\mathbf{1}_{[f(\mathbf{a},\mathbf{b}\odot u^{i+1},\mathbf{c})\neq f(\mathbf{a},\mathbf{b}\odot u^{i},\mathbf{c})]} \\
        &=\sum_{i=1}^{|s|}\underbrace{\left(\frac{1}{2^{n}} \displaystyle\sum_{\mathbf{a}\in\mathbb{B}^{|\mathcal{A}|}}
        \sum_{\mathbf{b}\in\mathbb{B}^{|\mathcal{B}|}} \sum_{\mathbf{c}\in\mathbb{B}^{|\mathcal{C}|}}\mathbf{1}_{[f(\mathbf{a},\mathbf{b}\odot u^{i}\odot u^{i+1},\mathbf{c})\neq f(\mathbf{a},\mathbf{b},\mathbf{c})]}\right)}_{=\text{The influence of the unique variable $b_{j(i)}$ that takes value $-1$ in $u^{i+1}\odot u^{i}$}}\\
        &< \eps/10.\quad [\text{Since every $b_j\in\mathcal{B}$ has influence $<\frac{\eps}{10k}$ and $|\mathcal{B}|\leq k$}]
    \end{align*}
    \end{proof}

    For each $\mathbf{a}\in\mathbb{B}^{|\mathcal{A}|}$, consider a fixed setting of strings $\mathbf{b}^{\mathbf{a}}\in\mathbb{B}^{|\mathcal{B}|}$,
    $\mathbf{c}^{\mathbf{a}}\in\mathbb{B}^{|\mathcal{C}|}$. Let us call the list of all these assignments $\Gamma$, i.e.
    $\Gamma=\{\forall\mathbf{a}\in\mathbb{B}^{|\mathcal{A}|}, (\mathbf{a}, \mathbf{b}^{\mathbf{a}}, \mathbf{c}^{\mathbf{a}})\}.$
    For any such ``list of assignments'' $\Gamma$, we define the function $F_{\Gamma}\colon\bits^n\to\bits$ as follows:
    $F_{\Gamma}(\mathbf{a},\ast,\ast)= f(\mathbf{a}, \mathbf{b}^{\mathbf{a}}, \mathbf{c}^{\mathbf{a}})$.
    The error incurred by approximating $f$ by $F_{\Gamma}$ is:
    \[  \P_{(\mathbf{a},\mathbf{b},\mathbf{c})}[F_{\Gamma}(\mathbf{a},\mathbf{b},\mathbf{c})\neq f(\mathbf{a},\mathbf{b},\mathbf{c})]
        =\P_{(\mathbf{a},\mathbf{b},\mathbf{c})}[f(\mathbf{a},\mathbf{b}^{\mathbf{a}},\mathbf{c}^{\mathbf{a}})\neq f(\mathbf{a},\mathbf{b},\mathbf{c})]
    \]
    \[=\P_{(\mathbf{a},\mathbf{b},\mathbf{c})}[f(\mathbf{a},\mathbf{b}^{\mathbf{a}},\mathbf{c})\neq f(\mathbf{a},\mathbf{b},\mathbf{c})]\quad
        [\text{Since $f$ does not depend on the variables in $\mathcal{C}$}]
    \]
    \begin{equation}
        =\frac{1}{2^{n}} \displaystyle\sum_{\mathbf{a}\in\mathbb{B}^{|\mathcal{A}|}}
        \sum_{\mathbf{b}\in\mathbb{B}^{|\mathcal{B}|}} \sum_{\mathbf{c}\in\mathbb{B}^{|\mathcal{C}|}}\mathbf{1}_{[f(\mathbf{a},\mathbf{b}^{\mathbf{a}},\mathbf{c})\neq f(\mathbf{a},\mathbf{b},\mathbf{c})]}
        =\frac{1}{2^{n}} \displaystyle\sum_{\mathbf{a}\in\mathbb{B}^{|\mathcal{A}|}} \sum_{s\in\mathbb{B}^{|\mathcal{B}|}} \sum_{\mathbf{c}\in\mathbb{B}^{|\mathcal{C}|}}\mathbf{1}_{[f(\mathbf{a},\mathbf{b}^{\mathbf{a}},\mathbf{c})\neq f(\mathbf{a},\mathbf{b}^{\mathbf{a}}\odot s,\mathbf{c})]}\label{eqn:fgamma}
    \end{equation}

    Therefore if we consider the expected value of the incurred error $\P[F_{\Gamma}\neq f]$ over all ``lists of assignments'' $\Gamma$, equation~\eqref{eqn:fgamma}
    implies that:
    \begin{align*}
        \E_{\Gamma}[\P_{(\mathbf{a},\mathbf{b},\mathbf{c})}[F_{\Gamma}\neq f]]&=\frac{1}{2^{|\mathcal{B}|}}\sum_{s\in\mathbb{B}^{|\mathcal{B}|}}\underbrace{\left(\frac{1}{2^{n}} \displaystyle\sum_{\mathbf{a}\in\mathbb{B}^{|\mathcal{A}|}}
        \sum_{\mathbf{b}^{\mathbf{a}}\in\mathbb{B}^{|\mathcal{B}|}} \sum_{\mathbf{c}\in\mathbb{B}^{|\mathcal{C}|}}\mathbf{1}_{[f(\mathbf{a},\mathbf{b}^{\mathbf{a}}\odot s,\mathbf{c})\neq f(\mathbf{a},\mathbf{b}^{\mathbf{a}},\mathbf{c})]}\right)}_{< \eps/10,\ \text{due to Fact~\ref{fac:stupidfact}}}\\
        &< \eps/10.
    \end{align*}
    Consequently, the expected error of approximating $f$ by a uniformly chosen $F_{\Gamma}$ is less than $\eps/10$. This also implies that for a
    uniformly chosen subset $\mathcal{S}$ of assignments to variables in $\mathcal{A}$ with size $(1-\eps/3)2^{|\mathcal{A}|}$,
    the expected error over $\mathcal{S}$ satisfies: $\E_{\Gamma}[\P_{\begin{subarray}{c}(\mathbf{a},\mathbf{b},\mathbf{c})\\
    \mathbf{a}\in \mathcal{S}
\end{subarray}}[F_{\Gamma}\neq f]]< \eps/10$. Therefore by Markov's Inequality, we obtain the following observation:
\begin{observation}\label{obs:jlproof}For a uniformly chosen subset $\mathcal{S}$ and $F_{\Gamma}$
as described above, $F_{\Gamma}$ will agree with $f$ on $(1-\eps/3)$ fraction of the coordinates
$\{(\mathbf{a},\mathbf{b},\mathbf{c}), \mathbf{a}\in \mathcal{S}\}$ with probability at least $7/10$.
\end{observation}

Now if we go back and recall what the algorithm does in Stage 2, we will observe that the generation of the hypothesis
in Stage 2 is equivalent to drawing a uniform $F_{\Gamma}$ and $\mathcal{S}$ as described and resetting the values of
$F_{\Gamma}$ at those coordinates $\{(\mathbf{a},\mathbf{b},\mathbf{c}), \mathbf{a}\notin \mathcal{S}\}$ to
\textsc{True}. This is because the algorithm only draws classical random examples during Stage 2. Therefore due to
Observation~\ref{obs:jlproof}, the hypothesis will disagree with $f$ on at most
$$\underbrace{1-(1-\eps/3)^{2}}_{\text{The error incurred by $(\mathbf{a},\mathbf{b},\mathbf{c}), \mathbf{a}\in \mathcal{S}$}}+\underbrace{\eps/3}_{\text{The error incurred by $(\mathbf{a},\mathbf{b},\mathbf{c}), \mathbf{a}\notin \mathcal{S}$}}< \eps$$ fraction of the inputs with overall probability at least $2/3$. This gives the desired result.
\end{proof}
Note that this algorithm
\begin{itemize}
    \item uses only a moderate number of quantum examples;
    \item has overall query complexity with no dependence on $n$, in contrast with known lower bounds
    (Lemma~\ref{S2Lem1}) for learning from classical membership
    queries;
    \item uses the $\ex$ oracle as its only source of classical information ($\mq$ queries are not
   used); and
    \item is computationally efficient.
\end{itemize}
One can compare this result to that of Theorem~\ref{S2Thm2} which
requires $\tilde{O}(n 2^{6k} \epsilon^{-8})$ quantum examples to
learn $k$-juntas. In contrast, our algorithm uses not only
substantially fewer quantum examples but also fewer uniform random
examples, which are considered quite cheap. Intuitively, this means
that for the junta learning problem, almost all the quantum queries
used by the algorithm of Bshouty and Jackson \cite{BSHJA} can in
fact be converted into ordinary classical random examples.

\subsubsection{Lower bounds}

The algorithm of Theorem~\ref{S2Thm3} is optimal in the following sense:
\begin{observation}\label{obsjlrn} Any $1/10$-learning quantum membership query algorithm for $k$-juntas
that uses only $\frac{1}{101} 2^k$ classical $\mq$ queries must
additionally use $\Omega(2^{k})$ $\qmq$ queries.
\end{observation}
\begin{proof}
This statement easily follows from Fact~\ref{S2Fac1} since a
classical membership query can be simulated by a $\qmq$ query.
\end{proof}
Contrasting our junta learning algorithm with Observation~\ref{obsjlrn}, we see that
if the allowed number of classical examples or queries is decreased
even slightly from the $O(2^{k} \log \eps^{-1})$ used by
our algorithm to $\frac{1}{101} 2^k$, then an additional
$\Omega(2^{k})$ quantum queries are required, even if  $\qmq$
queries are allowed.

\section{Conclusion}

We have given some results on learning and testing $k$-juntas using
both quantum examples and classical random examples.  It would be
interesting to develop other testing and learning algorithms that
combine these two sorts of oracles, with the goal of minimizing the
number of quantum oracle calls required.

Another interesting goal for future work is to further explore the
power of the $\bj$ oracle.  Can the gap between our
$O(k/\eps)$-query upper bound and our $\Omega(\sqrt{k})$-query lower
bound for the $\bj$ oracle be closed?


\begin{thebibliography}{99}

\bibitem{AKMPY}
A. Ambainis, K. Iwama, A. Kawachi, H. Masuda, R.~H. Putra, S.
Yamashita, \emph{Quantum Identification of Boolean Oracles},
Proceedings of STACS~2004, pp.~93-104.

\bibitem{AR}
    J. Arpe and R. Reischuk,
    \emph{Robust Inference of Relevant Attributes},
    Proceedings of the 14th International Conference on Algorithmic Learning Theory,
    pp.~99--113 (2003).

\bibitem{AR2}
    J. Arpe and R. Reischuk,
    \emph{Learning Juntas in the Presence of Noise},
    Proceedings of the 3rd International Conference on Theory and Applications
    of Models of Computation,
    pp.~387--398 (2006).

\bibitem{AS05}
    A. At{\i}c{\i}, R.~A. Servedio,
    \emph{Improved Bounds on Quantum Learning Algorithms},
    Quantum Information Processing, Vol.~\textbf{4}, No.~5, pp.~355--386 (2005).

\bibitem{Blum}
    A. Blum,
    \emph{Learning a Function of $r$ Relevant Variables (Open Problem)},
    Proceedings of the 16th Annual Conference on Learning Theory and 7th Kernel Workshop,
    pp.~731--733 (2003).
\bibitem{BFNR}
    H. Buhrman, L. Fortnow, I. Newman, H. R\"{o}hrig,
    \emph{Quantum Property Testing},
    Proceedings of 14th SODA, pp.~480--488 (2003).
\bibitem{BCG+96} N. Bshouty, R. Cleve, R. Gavald\`{a}, S. Kannan and
C. Tamon.
\emph{Oracles and queries that are sufficient for exact
learning},
J. Comput. Syst. Sci., Vol~\textbf{52}, No.~3 , pp. 421-433 (1996).

\bibitem{BSHJA}
    N.~H. Bshouty, J.~C. Jackson,
    \emph{Learning DNF over the Uniform Distribution Using a Quantum Example Oracle},
    SIAM J. Comput. Vol.~\textbf{28}, No.~3,
    pp.~1136--1153 (1999).

\bibitem{BV97}
    E. Bernstein, U. Vazirani,
    \emph{Quantum Complexity Theory},
    SIAM Journal of Computing, \textbf{26}(5): pp.~1411--1473 (1997).

\bibitem{C06}
    J. Castro,
    \emph{How many query superpositions are needed to learn?}
    Proceedings of 17th ALT, pp. 78-92 (2006).

\bibitem{CG04}
    H. Chockler, D. Gutfreund,
    \emph{A Lower Bound for Testing Juntas},
    Information Processing Letters \textbf{90}(6): pp.~301--305 (2004).

\bibitem{FKRSS}
    E. Fischer, G. Kindler, D. Ron, S. Safra, A. Samorodnitsky,
    \emph{Testing Juntas},
    Proceedings of the 43rd IEEE Symposium on Foundations of Computer Science,
    pp.~103--112 (2002).

\bibitem{FMSS}
    K. Friedl, F. Magniez, M. Santha, P. Sen.
    \emph{Quantum Testers for Hidden Group Properties},
    Proceedings of the 28th International Symposium on Mathematical Foundations
    of Computer Science, pp.~419--428.

\bibitem{GGR}
    O. Goldreich, S. Goldwasser, D. Ron,
    \emph{Property Testing and Its Connection to Learning and Approximation},
    Journal of the ACM, \textbf{45}(4): pp.~653--750 (1998).

\bibitem{HMPPR}
M. Hunziker, D.~A. Meyer, J. Park, J. Pommersheim and M. Rothstein,
\emph{The Geometry of Quantum Learning}, arXiv:quant-ph/0309059; to
appear in Quantum Information Processing.

\bibitem{IKRY}
K. Iwama, A. Kawachi, R. Raymond and S. Yamashita, \emph{Robust
Quantum Algorithms for Oracle Identification},
arXiv:quant-ph/0411204 (2005).


\bibitem{JACKSON}
        J.~C. Jackson,
        \emph{An Efficient Membership-Query Algorithm for Learning $\dnf$ with Respect to the Uniform Distribution},
        Journal of Computer and System Sciences \textbf{55}(3): 414--440 (1997).

\bibitem{KKL}
    J. Kahn, G. Kalai, N. Linial,
    \emph{The influence of variables on boolean functions},
    Proceedings of the 29th IEEE Symposium on Foundations of Computer Science,
    pp.~68--80 (1988).

\bibitem{KL}
    J. K\"{o}bler, W. Lindner,
    \emph{Learning Boolean Functions under the uniform distribution
    via the Fourier Transform},
    Bulletin of the EATCS 89 (2006).

\bibitem{KM}
        E. Kushilevitz, Y. Mansour,
        \emph{Learning Decision Trees using the Fourier Spectrum},
        SIAM Journal on Computing \textbf{22}(6): 1331-1348 (1993).

\bibitem{LMMV}
        R. Lipton, E. Markakis, A. Mehta, N. Vishnoi,
        \emph{On the Fourier Spectrum of Symmetric Boolean Functions with Applications to Learning Symmetric Juntas},
        Proceedings of the  20th Annual IEEE Conference on Computational Complexity, pp.~112--119 (2005).


\bibitem{MN}
    F. Magniez, A. Nayak.
    \emph{Quantum Complexity of Testing Group Commutativity},
    Proceedings of the 32nd International Colloquium on Automata,
    Languages and Programming, pp.~1312--1324 (2005).

\bibitem{MANSOUR}
    Y. Mansour.
    \emph{Learning {B}oolean functions via the {F}ourier transform},
    in ``Theoretical Advances in Neural Computation and
  Learning,'' Kluwer Academic Publishers, pp. 391-424 (1994).

\bibitem{MOS04}
    E. Mossel, R. O'Donnell and R.~A. Servedio,
    \emph{Learning Functions of $k$ Variables},
    Journal of Computer and System Sciences, Vol.~\textbf{69}, No.~3,
    pp.~421--434 (2004).
\bibitem{NC}
    M. Nielsen and I. Chuang,
    \emph{Quantum Computation and Quantum Information},
    Cambridge University Press (2000).
\bibitem{OS03}
    R. O'Donnell, R.~A. Servedio,
    \emph{Extremal Properties of Polynomial Threshold Functions},
    Journal of Computer \& System Sciences, to appear.  Available at
    http://www.cs.columbia.edu/\verb+~+rocco/papers/ccc03.html.  Preliminary
    version appeared in Eighteenth Annual IEEE Conference on Computational Complexity,
    pp.~3--12 (2003).
\bibitem{RS96}
    R. Rubinfeld and M. Sudan,
    \emph{Robust Characterizations of Polynomials with Applications to Program Testing},
    SIAM Journal on Computing, \textbf{25}(2): pp.~252--271 (1996).


\bibitem{RSSG}
R.~A. Servedio, S.~J. Gortler, \emph{Equivalences and Separations
between Quantum and Classical Learnability}, SIAM J. Comput.
Vol.~\textbf{33}, No. 5, pp.~1067--1092 (2004).



\bibitem{Val84}
    L.~G. Valiant,
    \emph{A Theory of the Learnable},
    Communications of the Association for Computing Machinery 27:11,
    pp.~1134--1142 (1984).
\end{thebibliography}
\end{document}